\documentclass[11pt,a4paper]{article}

\usepackage{mathtools}
\usepackage{authblk} 
\usepackage{algpseudocode,algorithmicx,algorithm}

\usepackage{mathrsfs}
\usepackage{latexsym,bm}
\usepackage{amsmath,amsfonts,amssymb,amsthm}
\usepackage{extarrows}

\usepackage{graphicx,subfigure,epstopdf,float}
\usepackage{enumerate,cases,multirow}
\usepackage{makecell}
\usepackage{caption}

\usepackage{longtable,colortbl,arydshln,threeparttable}
\definecolor{mygray}{gray}{.9}

\usepackage{indentfirst}
\usepackage[top=25mm,bottom=20mm,left=25mm,right=20mm]{geometry}
\usepackage{cite}

\usepackage{listings}

\usepackage{makeidx}        
\usepackage{booktabs}
\usepackage[bookmarks,bookmarksnumbered,colorlinks,citecolor=red,linkcolor=red,hyperindex,linktocpage=true]{hyperref}

\newcommand{\ket}[1]{| #1 \rangle} 
\newcommand{\bra}[1]{\langle #1 |} 

\newcommand{\bb}{\boldsymbol}

\def \d {\mathrm{d}}
\def \e {\mathrm{e}}
\def \i {\mathrm{i}}

\DeclareMathOperator{\diag}{diag}

\newcounter{parentalgorithm}

\makeatother

\newtheorem{theorem}{Theorem}[section]
\newtheorem{lemma}{Lemma}[section]

\theoremstyle{remark}
\newtheorem{remark}{\bf Remark}[section]

\numberwithin{equation}{section}

\usepackage{qcircuit}
\usepackage{placeins}

\begin{document}

\title{\bfseries Schr\"odingerization for quantum linear systems problems \thanks{Submitted to the editors \today. The authors are listed in alphabetical order. All authors contributed equally to this work and share the same co-first authorship.}}

\author{Yin Yang\thanks{yangyinxtu@xtu.edu.cn} \footnote{Corresponding author.}}
\author{Yue Yu\thanks{terenceyuyue@xtu.edu.cn} \footnote{Corresponding author.}}
\author{Long Zhang\thanks{longzhang@smail.xtu.edu.cn}}
	
\affil{Hunan Key Laboratory for Computation and Simulation in Science and Engineering, Key Laboratory of Intelligent Computing and Information Processing of Ministry of Education, National Center for Applied Mathematics in Hunan, School of Mathematics and Computational Science, Xiangtan University, Xiangtan, Hunan 411105, China}

\maketitle

\begin{abstract}
We develop a Schr\"odingerization algorithm for quantum linear systems
problems in two higher dimensions. The earlier LC-Schr\"odingerization
approach represents the solution as a time integral of a homogeneous
convection response and implements this integral by a linear combination
of unitaries (LCU) over evolution times. We instead use Duhamel's
principle to incorporate the integral into an inhomogeneous convection
equation with zero initial data. Fourier projection in the convection
variable and Schr\"odingerization in a second auxiliary variable then
realize the solution without the separate LCU step. A joint choice of
the kernel and recovery procedure gives an evolution time independent
of the target accuracy, a uniformly bounded $L^2$ kernel normalization,
and accurate recovery on a fixed interval. We establish the
periodization and discretization bounds and analyze the interval
recovery probability. With block preconditioning, we retain a single
logarithmic precision factor and obtain linear condition-number
dependence without variable-time amplitude amplification. Under exact
oracle access and given a valid constant-factor solution-norm estimate,
the algorithm uses $\mathcal O(\kappa_A\log(1/\varepsilon))$ queries
to each original input oracle, with constant success probability and
$\ell^2$ state error at most $\varepsilon$. The matrix-query bound
matches the standard worst-case scaling when the supplied norm bounds
are tight. UnitaryLab simulations on positive-definite and indefinite
systems demonstrate the solution accuracy and probability gain from
interval recovery.
\end{abstract}


	
\section{Introduction}
Quantum computing offers a different approach to computation, with
algorithms that can substantially outperform classical methods for
certain problems \cite{Nielsen2010,HHL2009}. These advances
have motivated the development of quantum algorithms for scientific
computing. Solving linear systems is a fundamental task in this
setting, with applications across science and engineering. The cost
of classical solvers grows with the number of unknowns and can become
particularly demanding for systems obtained by discretizing
high-dimensional partial differential equations. Quantum linear
systems algorithms (QLSAs) approach this task by preparing a quantum
state proportional to the solution.

We consider the quantum linear systems problem (QLSP) for
\begin{equation}\label{Linearsystem}
A\bb x=\bb b,
\end{equation}
where $A\in\mathbb C^{N\times N}$ is invertible, $N=2^n$, and
$\bb b\ne\bb0$. Given a block encoding of $A$ and a preparation
oracle for $\ket b=\bb b/\|\bb b\|$, we seek a state
$\ket{\widetilde x}$ satisfying
\[
\ket x=\frac{A^{-1}\bb b}{\|A^{-1}\bb b\|},\qquad
\|\ket{\widetilde x}-\ket x\|\le\varepsilon,
\]
with a success flag of constant probability. The potential quantum
advantage in the system dimension relies on efficient input access,
favorable conditioning and accuracy requirements, and an application
that uses the solution state or a small number of its observables
without reconstructing every component of the solution
\cite{HHL2009,Lin2026QLSASurvey}.

We use the supplied bounds
$\alpha_A\ge\|A\|$ and $\alpha_{A^{-1}}\ge\|A^{-1}\|$, and write
$\kappa_A=\alpha_A\alpha_{A^{-1}}$.
Here $\alpha_A$ is the matrix block-encoding normalization.
We denote the solution norm for a normalized right-hand side by
\[
\xi=\|A^{-1}\ket b\|,\qquad
\frac1{\alpha_A}\le\xi\le\alpha_{A^{-1}}.
\]
The lower bound follows from $1\le\|A\|\xi$, and implies $\kappa_A\ge1$.
It suffices to consider Hermitian matrices with either sign of the
spectrum, since standard Hermitian dilation reduces a general
invertible system to this case. The precise oracle assumptions are stated in
Section~\ref{sec:algorithm}.

The algorithm of Harrow, Hassidim, and Lloyd~\cite{HHL2009} uses
phase estimation and eigenvalue inversion to prepare the solution
state. Its polynomial dependence on the condition number and inverse
precision motivated subsequent improvements. Ambainis introduced
variable-time amplitude amplification (VTAA) to obtain a nearly
linear dependence on the condition number~\cite{Ambainis2012VTAA}.
Childs, Kothari, and Somma replaced phase-estimation-based inversion
with Fourier or Chebyshev approximations, obtaining polylogarithmic
precision dependence and, with VTAA, nearly linear condition-number
dependence~\cite{Childs2017QLSA}. Quantum singular value transformation
provides a related approach through polynomial transformations of
block-encoded matrices~\cite{Gilyen2019QSVD}.

Another line of work prepares the solution through adiabatic
evolution, randomization, and eigenstate filtering
\cite{Subasi2019AQC,An-Lin-2022,Lin-Tong-2020}.
Costa et al.~\cite{Costa2022QLSA} use a discrete adiabatic theorem
to achieve linear condition-number and logarithmic precision
dependence without VTAA. More recent developments include
block preconditioning, which controls the costs of matrix access
and right-hand-side preparation~\cite{Low2026quantumlinearsystem},
and a solver based on solution-norm estimation and kernel reflection
\cite{dalzell2026shortcutoptimalquantumlinear}.
We refer to \cite{Lin2026QLSASurvey} for a broader account.
These developments provide the complexity benchmarks for the
evolution-based realization studied here.

An evolution equation provides another way to represent the solution.
For positive-definite Hermitian $A$, the relaxation equation
$\d\bb z/\d t=-A\bb z+\bb b$ approaches $A^{-1}\bb b$
\cite{HJZ2024multiscale}.
Schr\"odingerization converts non-unitary linear dynamics into a
Hamiltonian system in an enlarged space
\cite{JLY22SchrShort,JLY22SchrLong,JLMPY2025Schropt}.
The periodic initialization and weighted recovery developed in
\cite{DMPYSchrodingerizationODE} give an autonomous dissipative solver
that we can use for such a representation. However, the relaxation
equation above is unstable when $A$ has negative eigenvalues.
Replacing $A$ by $A^2$ restores decay, but squares the spectral
condition number. This motivates us to seek an evolution that treats
positive and negative eigenvalues directly.

The work introducing LC-Schr\"odingerization
\cite{yang2026linearcombinationschrodingerizationquantum} develops
an abstract framework that represents $A^{-1}$ through suitable
kernel functions and the definition of the Fourier transform.
The kernel used in \cite{Childs2017QLSA} provides one example within
this framework. The LC-Schr\"odingerization construction then expresses
the solution as a time integral of a homogeneous convection response.
For an appropriate odd kernel $\zeta$, let $\bb u$ solve
$\partial_t\bb u=A\partial_p\bb u$ with
$\bb u(0,p)=\zeta(p)\bb b$. Then
\[
A^{-1}\bb b=\lim_{T\to\infty}\int_0^T\bb u(t,0)\,\d t.
\]
Fourier discretization makes the convection evolution unitary, while
quadrature of the time integral is implemented by a linear combination
of unitaries (LCU) over the evolution times. The resulting algorithm
therefore treats propagation and integration as separate operations.
With a suitable kernel, interval recovery, and block preconditioning,
that approach already attains the optimal matrix-query scaling
$\mathcal O(\kappa_A\log(1/\varepsilon))$ under its oracle and
solution-norm estimation assumptions.

In this work, we develop a new realization of this inverse
representation that avoids the LCU step for time integration.
We incorporate the accumulated response into the dynamics by
considering
\[
\partial_t\bb v(t,p)=A\partial_p\bb v(t,p)+\zeta(p)\bb b,
\qquad \bb v(0,p)=\bb0.
\]
Duhamel's principle gives
$\bb v(T,p)=\int_0^T\bb u(t,p)\,\d t$, so the inhomogeneous
equation performs the required integration and $\bb v(T,0)$
approaches $A^{-1}\bb b$.
Both spectral signs are accommodated through the transport velocities.
After imposing periodic boundary conditions and using
Fourier projection in $p$, we obtain a finite inhomogeneous ODE.
We then apply the Schr\"odingerization method
\cite{JLY22SchrShort,JLY22SchrLong,JLMPY2025Schropt,DMPYSchrodingerizationODE},
introducing a second auxiliary variable $q$.
With the two auxiliary variables $p$ and $q$, we thus
realize Schr\"odingerization of the QLSP in two higher dimensions.
The inhomogeneous dynamics already contain the time-integrated response,
so no separate LCU step over evolution times is required.

The LCU-based construction already uses interval recovery to avoid
the probability loss from selecting a single grid point
\cite{yang2026linearcombinationschrodingerizationquantum}.
We retain this principle and adapt the kernel to the present
dynamical realization: we seek a response $\bb v(T,p)$ that is
uniformly close to $\bb x$ itself on a fixed interval, to any
prescribed accuracy.
The choice of kernel is essential: for the kernel
$\zeta(p)=p\e^{-p^2/2}$ used in the Fourier method
\cite{Childs2017QLSA}, the limiting solution is
$\e^{-p^2/2}\bb x$, so recovery of $\bb x$ to increasing accuracy
requires a shrinking neighbourhood of the origin.
We therefore choose the kernel and recovery procedure jointly.
Our smoothed-interval kernels have uniformly bounded $L^2$ norms,
and their tail integrals approach one uniformly on $[-1,1]$.
Adjusting the smoothing parameter allows us to improve the accuracy
while keeping both the recovery interval and the evolution time fixed.

Our main contributions are as follows.
\begin{itemize}
\item \textbf{Schr\"odingerization in two higher dimensions.}
The earlier LC-Schr\"odingerization method
\cite{yang2026linearcombinationschrodingerizationquantum} already
achieves optimal matrix-query scaling, but uses LCU to combine
homogeneous evolution states at different times.
We use Duhamel's principle to turn this time integral into an
inhomogeneous convection equation with zero initial data.
After Fourier projection in $p$, applying the Schr\"odingerization
method of \cite{DMPYSchrodingerizationODE} introduces a second
auxiliary variable $q$. We thus obtain a Schr\"odingerization
representation of the QLSP in two higher dimensions, avoiding the
separate LCU step for time integration.

\item \textbf{Joint choice of the kernel and recovery procedure.}
We construct a smoothed-interval kernel and recover the solution
coherently on the fixed interval $[-1,1]$. This choice gives
\[
T=5\alpha_{A^{-1}},\qquad \|\zeta\|_{L^2(\mathbb R)}=\mathcal O(1),
\]
independently of the target accuracy $\varepsilon$; increasing the
kernel parameter $r$ reduces the recovery error exponentially.
Theorems~\ref{thm:newkernel} and~\ref{thm:finite-time-recovery}
establish these properties. After discretization, the number of
accepted grid points compensates for the sampled kernel normalization,
so interval recovery introduces no additional probability loss from
grid refinement.

\item \textbf{Block preconditioning without VTAA.}
We apply the block preconditioner of \cite{Low2026quantumlinearsystem}
to reduce the dependence on $\kappa_A$ from quadratic to linear
without VTAA. The resulting algorithm retains the single factor of
$\log(1/\varepsilon)$ obtained from the kernel and recovery estimates
and the Schr\"odingerization solver.
Under the assumptions of Theorem~\ref{thm:optimal-queries}, the
algorithm uses
\[
\mathcal O\!\left(\kappa_A\log\frac1\varepsilon\right)
\]
queries to each original input oracle, with constant success
probability and $\ell^2$ state error at most $\varepsilon$.
This recovers the optimal matrix-query scaling within our
two-variable Schr\"odingerization realization.
\end{itemize}

Our query bounds assume exact matrix block encodings and exact
preparation of the input and auxiliary profile states, with the
required controlled and inverse access. The preconditioning step
uses a valid constant-factor estimate of $\xi$; we discuss this assumption in
Section~\ref{subsec:block-preconditioning}.
The supplied norm bounds determine $\kappa_A$, which agrees with
the spectral condition number when these bounds are tight.
We illustrate the construction with UnitaryLab circuits for
positive-definite and indefinite systems, examining solution errors,
discretization refinement, and recovery probabilities. The experiments
use a product formula for Hamiltonian evolution and test the
representation and recovery rather than the asymptotic query bound.

We organize the paper as follows. Section~\ref{sec:kernel} develops
the convection representation and the smoothed-interval kernel.
Section~\ref{sec:discretization} establishes periodization, Fourier
discretization, and interval recovery estimates.
Section~\ref{sec:algorithm} gives the Schr\"odingerization algorithm
and its query complexity, including block preconditioning.
Section~\ref{sec:numerics} presents the numerical experiments,
followed by the conclusions.

\paragraph{Notation.}
We use $\|\cdot\|$ for the Euclidean norm of a vector and the induced
operator norm of a matrix.
For functions, $\|\cdot\|_2$ and $\|\cdot\|_\infty$ abbreviate the
$L^2$ and $L^\infty$ norms on the stated domain; kernel norms without
an explicit domain are taken over $\mathbb R$.
The symbols $A^\top$ and $A^\dagger$ denote the transpose and conjugate
transpose, respectively; $I$ denotes the identity of the appropriate
dimension, and $|\mathcal I|$ denotes the cardinality of a finite set.
Unless stated otherwise, $\log$ denotes the natural logarithm.

For nonnegative quantities $f$ and $g$, we write $f=\mathcal O(g)$
if $f\le Cg$, $f=\Omega(g)$ if $f\ge cg$, and $f=\Theta(g)$ if
$cg\le f\le Cg$, for positive constants $c$ and $C$ in the parameter
regime under consideration. Unless stated otherwise, these constants
are independent of the problem size, accuracy, and other varying
parameters. We also use $c$ and $C$ for generic positive constants
whose values may change from one occurrence to another.

\section{Convection representation of quantum linear systems problems}\label{sec:kernel}
By standard Hermitian dilation, we may assume that $A$ is Hermitian.
Duhamel's principle leads to a
convection representation whose solution approximates $\bb x=A^{-1}\bb b$
near the origin. We then choose the kernel and recovery interval
jointly to keep the evolution time independent of the accuracy and
the kernel's $L^2$ norm uniformly bounded.

\subsection{Convection representation and local recovery}

Our first aim is to represent $\bb x$ by the solution of a convection
equation. We require convergence to $\bb x$ at the origin and, for a
given accuracy, an approximation to the same vector throughout a
neighbourhood of the origin. The following theorem gives both properties
and tells us how to choose the recovery interval and evolution time.

\begin{theorem}[Convection representation and approximate interval recovery]
\label{thm:discreteFouriersource}
Let $A$ be an invertible Hermitian matrix, $\bb b\ne\bb0$, and
$\bb x=A^{-1}\bb b$. Let $\zeta\in C^1(\mathbb R)\cap L^1(\mathbb R)$
be odd and satisfy $\int_0^\infty\zeta(s)\,\d s=1$. Define
\[
F(p)=\int_p^\infty\zeta(s)\,\d s.
\]
The solution of
\begin{equation}\label{vtpsource}
\begin{cases}
\partial_t\bb v(t,p)=A\partial_p\bb v(t,p)+\zeta(p)\bb b,\\
\bb v(0,p)=\bb0,
\end{cases}
\end{equation}
satisfies
\begin{equation}\label{integraltt}
\bb x=\lim_{T\to\infty}\bb v(T,0).
\end{equation}

Moreover, for $a>0$, $\alpha_{A^{-1}}\ge\|A^{-1}\|$, and
$T\ge a\alpha_{A^{-1}}$,
\begin{equation}\label{abstract-interval-error}
\sup_{|p|\le a}\|\bb v(T,p)-\bb x\|
\le\left(
\sup_{|p|\le a}|1-F(p)|
+\sup_{|s|\ge T/\alpha_{A^{-1}}-a}|F(s)|
\right)\|\bb x\|.
\end{equation}
In particular, for $0<\delta<1$, choose $a_\delta,B_\delta>0$ such that
\begin{equation}\label{recovery-radius-choice}
\sup_{|p|\le a_\delta}|1-F(p)|\le\frac\delta2,
\qquad
\sup_{|s|\ge B_\delta}|F(s)|\le\frac\delta2.
\end{equation}
Such choices always exist. With
\begin{equation}\label{recovery-time-choice}
T_\delta=\alpha_{A^{-1}}(a_\delta+B_\delta),
\end{equation}
we have
\[
\sup_{|p|\le a_\delta}\|\bb v(T,p)-\bb x\|
\le\delta\|\bb x\|\qquad(T\ge T_\delta).
\]
\end{theorem}
\begin{proof}
Integrability, oddness, and the half-line normalization imply that
$F(0)=1$, $F$ is even, and $F(p)\to0$ as $|p|\to\infty$.
To derive the solution formula, we introduce the auxiliary homogeneous
problem
\begin{equation}\label{homogeneous-response}
\begin{cases}
\partial_t\bb u(t,p)=A\partial_p\bb u(t,p),\\
\bb u(0,p)=\zeta(p)\bb b.
\end{cases}
\end{equation}
We write $A=\sum_j\lambda_j\bb e_j\bb e_j^\dagger$
in an orthonormal eigenbasis and set $b_j=\bb e_j^\dagger\bb b$.
Characteristics give
$u_j(t,p)=\zeta(p+\lambda_jt)b_j$ in the $j$th eigenspace,
so $\bb u(t,p)=\zeta(pI+tA)\bb b$.
By Duhamel's principle, the solution of \eqref{vtpsource} is
\[
\bb v(t,p)=\int_0^t\bb u(s,p)\,\d s.
\]
Indeed, differentiation under the integral gives
\[
A\partial_p\bb v(t,p)
=\int_0^t\partial_s\bb u(s,p)\,\d s
=\bb u(t,p)-\zeta(p)\bb b
=\partial_t\bb v(t,p)-\zeta(p)\bb b,
\]
and $\bb v(0,p)=\bb0$. Uniqueness follows by applying
characteristics to the difference of two solutions, which satisfies
the homogeneous equation with zero initial data.
Consequently, $v_j(T,p)=b_j\int_0^T\zeta(p+\lambda_jt)\,\d t$.
Since $F'=-\zeta$ and $\lambda_j\ne0$,
\[
v_j(T,p)=\frac{b_j}{\lambda_j}
\bigl[F(p)-F(p+\lambda_jT)\bigr].
\]
Summing over the eigenspaces gives the identity
\begin{equation}\label{convection-solution}
\bb v(T,p)=\int_0^T\zeta(pI+sA)\bb b\,\d s
=\bigl[F(p)I-F(pI+TA)\bigr]\bb x,
\end{equation}
valid for both signs of the spectrum. Setting $p=0$ and using $F(0)=1$
gives
\[
\|\bb v(T,0)-\bb x\|
\le\max_{\lambda\in\operatorname{spec}(A)}|F(T\lambda)|\,\|\bb x\|.
\]
Since every eigenvalue is nonzero and $F$ vanishes at both infinities,
this proves \eqref{integraltt}.

For interval recovery, the error is
\[
\bb v(T,p)-\bb x
=\bigl[(F(p)-1)I-F(pI+TA)\bigr]\bb x.
\]
To bound the two contributions, we set
\begin{equation}\label{abstract-recovery-errors}
\eta_a=\sup_{|p|\le a}|1-F(p)|,\qquad
\theta_a(T)=\sup_{|p|\le a}\max_{\lambda\in\operatorname{spec}(A)}
|F(p+T\lambda)|.
\end{equation}
The spectral theorem and the triangle inequality bound the uniform
error by $[\eta_a+\theta_a(T)]\|\bb x\|$.
Since $|p+T\lambda|\ge T/\alpha_{A^{-1}}-a$,
\[
\theta_a(T)\le\sup_{|s|\ge T/\alpha_{A^{-1}}-a}|F(s)|,
\]
which proves \eqref{abstract-interval-error}.
Finally, continuity and $F(0)=1$ guarantee the first choice in
\eqref{recovery-radius-choice}; decay of $F$ at both infinities
guarantees the second. For $T\ge T_\delta$ as defined in
\eqref{recovery-time-choice},
$T/\alpha_{A^{-1}}-a_\delta\ge B_\delta$.
Both terms in \eqref{abstract-interval-error} are then at most
$\delta/2$, proving the final assertion.
\end{proof}

\begin{remark}[Kernel choice for fixed-interval recovery]
\label{rem:interval-accuracy}
We need to choose the kernel so that we can recover $\bb x$ throughout
an interval of fixed positive width, independent of the target accuracy.
After discretization, retaining only $p=0$ loses success probability
as the grid is refined. Allowing the recovery interval to shrink with
the error tolerance also reduces the accepted probability and increases
the amplitude-amplification cost. A fixed interval lets us accept more
grid points as the mesh is refined, compensating for the normalization
of the sampled kernel.

The standard Gaussian choice $\zeta(p)=p\e^{-p^2/2}$ does not meet this
requirement. Its tail integral is $F(p)=\e^{-p^2/2}$, and
$\lim_{T\to\infty}\bb v(T,p)=F(p)\bb x$. On $[-a,a]$, the limiting
relative error is therefore $1-\e^{-a^2/2}$; making it at most $\delta$
requires $a=\mathcal O(\sqrt\delta)$. Increasing $T$ cannot remove
this spatial error. We therefore choose a different kernel below,
adjusting its smoothing parameter so that $F$ stays sufficiently close
to one on the fixed interval $[-1,1]$.
\end{remark}

\subsection{A smoothed-interval kernel}

We choose $[-1,1]$ as our recovery interval and seek a kernel whose
tail integral is nearly one there, as motivated by
Remark~\ref{rem:interval-accuracy}. The kernel must also have a
uniformly bounded $L^2$ norm as we improve the accuracy.
We construct it in three steps.
\begin{itemize}
\item \textbf{Initial profile.}
We start with the odd function
\[
f(p)=\mathbf{1}_{[2,3]}(p)-\mathbf{1}_{[-3,-2]}(p).
\]
It has unit positive-half-line mass, and its tail integral is exactly
one for $|p|\le2$. Our recovery interval therefore lies a unit distance
inside this constant region. The endpoints $2$ and $3$ are convenient
fixed choices, independent of the matrix and target accuracy.

\item \textbf{Gaussian smoothing.}
We smooth $f$ by convolution with the normalized Gaussian
\[
G_\sigma(p)=\frac{\e^{-p^2/(2\sigma^2)}}{\sqrt{2\pi}\sigma},
\qquad \sigma=(2r)^{-1/2},\quad r\ge2\ \text{an integer}.
\]
Writing the convolution explicitly, we obtain
\begin{equation}\label{kernel-convolution}
\begin{aligned}
\widetilde\zeta_\sigma(p)
&:=(G_\sigma*f)(p)
=\int_{\mathbb R}G_\sigma(p-y)f(y)\,\d y\\
&=\int_2^3G_\sigma(p-y)\,\d y
-\int_{-3}^{-2}G_\sigma(p-y)\,\d y\\
&=\int_2^3\bigl[G_\sigma(p-y)-G_\sigma(p+y)\bigr]\,\d y.
\end{aligned}
\end{equation}
The resulting function is smooth and has decaying tails. Gaussian leakage across the unit
separation is controlled by $\exp(-1/(2\sigma^2))=\e^{-r}$.

\item \textbf{Unit positive-half-line integral.}
Gaussian smoothing preserves oddness, but the integral of
$\widetilde\zeta_\sigma$ over $(0,\infty)$ is no longer exactly one.
We denote this integral by $m_\sigma$ and set
$\zeta=\widetilde\zeta_\sigma/m_\sigma$, so that
$\int_0^\infty\zeta(p)\,\d p=1$. This ensures $F(0)=1$ and hence
$\lim_{T\to\infty}\bb v(T,0)=\bb x$ in
Theorem~\ref{thm:discreteFouriersource}.
\end{itemize}
The positive-half-line integral has the explicit expression
\[
m_\sigma=\int_0^\infty\widetilde\zeta_\sigma(p)\,\d p
=\int_2^3\operatorname{erf}
\left(\frac{y}{\sqrt2\sigma}\right)\d y,
\]
where $\operatorname{erf}(z)=\frac2{\sqrt\pi}\int_0^z\e^{-t^2}\d t$.
We use the Fourier convention
\[
\hat\varphi(p)=\int_{\mathbb R}\varphi(k)\e^{-\i kp}\,\d k.
\]
With this convention, the kernel pair is
\begin{equation}\label{varphik}
\begin{aligned}
	\zeta(p)=\hat{\varphi}(p)
	&=\frac1{m_\sigma}\int_2^3
	\bigl(G_\sigma(p-y)-G_\sigma(p+y)\bigr)\d y,\\
	\varphi(k)&=\frac{\i}{\pi m_\sigma}\e^{-\sigma^2k^2/2}
	\int_2^3\sin(ky)\d y .
\end{aligned}
\end{equation}
The separation of the two signed intervals from the origin is what
allows the tail integral to remain nearly constant on $[-1,1]$.

\begin{figure}[tbp]
\centering
\subfigure[Source kernel.]{%
\includegraphics[width=0.32\textwidth]{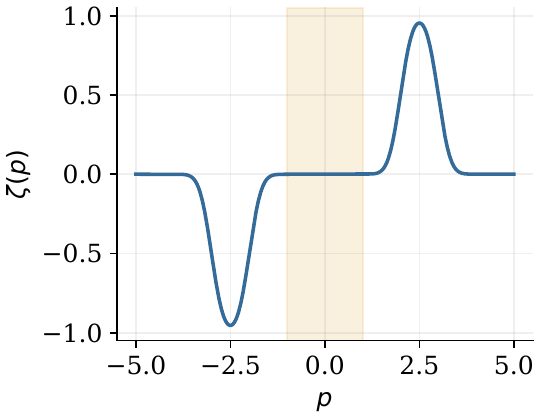}%
\label{fig:kernel-source}}\hfill
\subfigure[Fourier kernel.]{%
\includegraphics[width=0.32\textwidth]{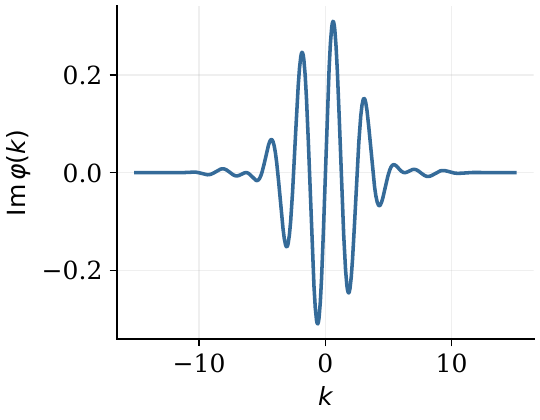}%
\label{fig:kernel-fourier}}\hfill
\subfigure[Recovery profile.]{%
\includegraphics[width=0.32\textwidth]{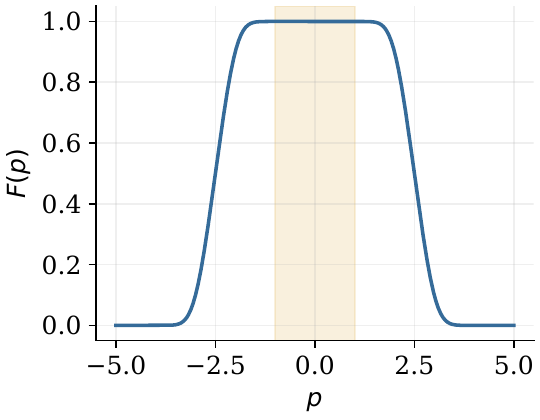}%
\label{fig:kernel-recovery}}
\caption{The Fourier kernel pair in \eqref{varphik} at the fixed illustrative
value $\sigma=1/4$ ($r=8$).
\subref{fig:kernel-source} $\zeta(p)=\hat\varphi(p)$.
\subref{fig:kernel-fourier} $\operatorname{Im}\varphi(k)$; the real part vanishes identically.
\subref{fig:kernel-recovery} The tail integral $F(p)$. Shading marks the recovery
interval $[-1,1]$, where $F$ is nearly one. Gaussian smoothing gives
decaying tails, rather than compact support.}
\label{fig:kernel-pair}
\end{figure}

The integer $r$ controls the smoothing width and the recovery accuracy. The next theorem bounds the variation and decay of $F$ and records the norm and regularity of the kernel. In Section~\ref{sec:discretization}, we use its explicit Fourier coefficients to control the projection error.

\begin{theorem}[Properties of the smoothed-interval kernel]\label{thm:newkernel}
Let $r\ge2$ be an integer, $\sigma=(2r)^{-1/2}$, and let
$\zeta$ and $\varphi$ be defined by \eqref{varphik}. For their tail
integral $F$ in Theorem~\ref{thm:discreteFouriersource} and the constant
$c_G=\operatorname{erf}(1/\sqrt2)$, the following properties hold.
\begin{itemize}
\item The functions $\zeta$ and $\varphi$ form a Fourier transform pair, and
$\zeta\in C^\infty(\mathbb R)\cap L^1(\mathbb R)$.
The function $\zeta$ is odd, positive on $(0,\infty)$,
and satisfies
\[
\int_0^\infty\zeta(p)\d p=1,\qquad
c_G\le m_\sigma\le1.
\]
\item The function $F$ satisfies
\begin{equation}\label{newkernel-recovery}
0\le1-F(p)\le\frac{\e^{-r}}{2c_G},\qquad |p|\le1,
\end{equation}
and
\begin{equation}\label{newkernel-tail}
	0\le F(p)
	\le \frac1{2c_G}\exp\left(-\frac{(|p|-3)^2}{2\sigma^2}\right),
\qquad |p|\ge3.
\end{equation}
\item For every integer $k\ge1$,
\begin{equation}\label{newkernel-derivatives}
	\|\zeta\|_{L^2}\le\frac{\sqrt2}{c_G},\qquad
	\|\zeta^{(k)}\|_{L^2}
	\le\frac{\sqrt2}{c_G}
	\left(\frac{\sqrt{k}}{\sigma\sqrt{\e}}\right)^k .
\end{equation}
\end{itemize}
All constants are independent of $r$.
\end{theorem}
\begin{proof}
We first verify the normalization and the Fourier transform pair.
For $y>0$, a change of variables gives
\begin{align*}
\int_0^\infty\bigl[G_\sigma(p-y)-G_\sigma(p+y)\bigr]\,\d p
&=\int_{-y}^y G_\sigma(s)\,\d s\\
&=\frac2{\sqrt\pi}\int_0^{y/(\sqrt2\sigma)}\e^{-u^2}\,\d u
=\operatorname{erf}\left(\frac{y}{\sqrt2\sigma}\right).
\end{align*}
Since $\sigma\le1$ and $2\le y\le3$, we obtain
\[
c_G\le\operatorname{erf}\left(\frac{y}{\sqrt2\sigma}\right)\le1,
\qquad
c_G\le m_\sigma=\int_2^3\operatorname{erf}
\left(\frac{y}{\sqrt2\sigma}\right)\,\d y\le1.
\]
Consequently,
\[
\int_0^\infty\zeta(p)\,\d p
=\frac1{m_\sigma}\int_2^3\operatorname{erf}
\left(\frac{y}{\sqrt2\sigma}\right)\,\d y=1.
\]
The evenness of $G_\sigma$ and its strict decrease on $(0,\infty)$ give
\[
\zeta(-p)=-\zeta(p),\qquad
G_\sigma(p-y)>G_\sigma(p+y)\quad(p,y>0),
\]
so $\zeta(p)>0$ for $p>0$.
Differentiating under the integral on the bounded interval $[2,3]$,
we have, for every integer $k\ge0$,
\[
\zeta^{(k)}(p)=\frac1{m_\sigma}\int_2^3
\bigl[G_\sigma^{(k)}(p-y)-G_\sigma^{(k)}(p+y)\bigr]\,\d y.
\]
Thus $\zeta\in C^\infty(\mathbb R)$. Using
$\zeta=m_\sigma^{-1}G_\sigma*f$ and Young's inequality, we also obtain
\[
\|\zeta\|_{L^1}
\le\frac{\|G_\sigma\|_{L^1}\|f\|_{L^1}}{m_\sigma}
=\frac2{m_\sigma}<\infty.
\]
For the Fourier transform, the Gaussian integral yields
\[
\frac1{2\pi}\int_{\mathbb R}G_\sigma(p-y)\e^{\i\omega p}\,\d p
=\frac1{2\pi}\e^{\i\omega y}\e^{-\sigma^2\omega^2/2}.
\]
Applying this identity to the two translates in $\zeta$, we find
\begin{align*}
\frac1{2\pi}\int_{\mathbb R}\zeta(p)\e^{\i\omega p}\,\d p
&=\frac{\e^{-\sigma^2\omega^2/2}}{2\pi m_\sigma}
\int_2^3\bigl(\e^{\i\omega y}-\e^{-\i\omega y}\bigr)\,\d y\\
&=\frac{\i}{\pi m_\sigma}\e^{-\sigma^2\omega^2/2}
\int_2^3\sin(\omega y)\,\d y=\varphi(\omega),
\end{align*}
which proves the Fourier transform relation in \eqref{varphik}.

We next estimate $F$. For $z\ge0$, substitution $u=z+v$ gives
\begin{align*}
\frac1{\sqrt{2\pi}}\int_z^\infty\e^{-u^2/2}\,\d u
&=\frac{\e^{-z^2/2}}{\sqrt{2\pi}}
\int_0^\infty\e^{-zv-v^2/2}\,\d v\\
&\le\frac{\e^{-z^2/2}}{\sqrt{2\pi}}
\int_0^\infty\e^{-v^2/2}\,\d v
=\frac12\e^{-z^2/2}.
\end{align*}
By oddness and integrability of $\zeta$, its integral over $\mathbb R$
vanishes, and hence
\[
F(-p)=-\int_{-\infty}^{-p}\zeta(s)\,\d s
=\int_p^\infty\zeta(s)\,\d s=F(p).
\]
It therefore suffices to consider $p\ge0$.
For $0\le p\le1$, we use $F(0)=1$, positivity of $\zeta$, and the
Gaussian tail bound to obtain
\begin{align*}
0\le1-F(p)
&=\frac1{m_\sigma}\int_2^3\int_0^p
\bigl[G_\sigma(s-y)-G_\sigma(s+y)\bigr]\,\d s\,\d y\\
&\le\frac1{m_\sigma}\int_2^3\int_{-\infty}^1
G_\sigma(s-y)\,\d s\,\d y\\
&=\frac1{m_\sigma}\int_2^3\frac1{\sqrt{2\pi}}
\int_{(y-1)/\sigma}^\infty\e^{-u^2/2}\,\d u\,\d y\\
&\le\frac1{2m_\sigma}\int_2^3
\e^{-(y-1)^2/(2\sigma^2)}\,\d y
\le\frac{\e^{-1/(2\sigma^2)}}{2c_G}
=\frac{\e^{-r}}{2c_G}.
\end{align*}
Together with the evenness of $F$, this proves \eqref{newkernel-recovery}.
For $p\ge3$, the same tail estimate gives
\begin{align*}
0\le F(p)
&=\frac1{m_\sigma}\int_2^3\int_p^\infty
\bigl[G_\sigma(s-y)-G_\sigma(s+y)\bigr]\,\d s\,\d y\\
&\le\frac1{m_\sigma}\int_2^3\frac1{\sqrt{2\pi}}
\int_{(p-y)/\sigma}^\infty\e^{-u^2/2}\,\d u\,\d y\\
&\le\frac1{2m_\sigma}\int_2^3
\e^{-(p-y)^2/(2\sigma^2)}\,\d y
\le\frac1{2c_G}\e^{-(p-3)^2/(2\sigma^2)}.
\end{align*}
Evenness then gives \eqref{newkernel-tail} for $|p|\ge3$.

Finally, we derive the norm and derivative estimates explicitly.
Writing the inverse Fourier transform of the initial profile as
\[
h_f(\omega):=\frac1{2\pi}\int_{\mathbb R}f(p)\e^{\i\omega p}\,\d p
=\frac{\i}{\pi}\int_2^3\sin(\omega y)\,\d y.
\]
The convolution formula and Plancherel's identity give
\[
\varphi(\omega)=\frac1{m_\sigma}\e^{-\sigma^2\omega^2/2}h_f(\omega),
\qquad
2\pi\int_{\mathbb R}|h_f(\omega)|^2\,\d\omega
=\|f\|_{L^2}^2=2.
\]
In particular,
\[
\|\zeta\|_{L^2}^2
=\frac{2\pi}{m_\sigma^2}\int_{\mathbb R}
\e^{-\sigma^2\omega^2}|h_f(\omega)|^2\,\d\omega
\le\frac2{m_\sigma^2}\le\frac2{c_G^2}.
\]
For each integer $k\ge1$, the Gaussian factor makes
$\omega^k\varphi(\omega)$ integrable and square integrable, so we may
differentiate the Fourier representation to obtain
\[
\zeta^{(k)}(p)=\int_{\mathbb R}(-\i\omega)^k
\varphi(\omega)\e^{-\i\omega p}\,\d\omega.
\]
Applying Plancherel's identity once more, we find
\begin{align*}
\|\zeta^{(k)}\|_{L^2}^2
&=\frac{2\pi}{m_\sigma^2}\int_{\mathbb R}
|\omega|^{2k}\e^{-\sigma^2\omega^2}|h_f(\omega)|^2\,\d\omega\\
&\le\frac{2\pi}{m_\sigma^2}
\left(\sup_{u\ge0}u^{2k}\e^{-\sigma^2u^2}\right)
\int_{\mathbb R}|h_f(\omega)|^2\,\d\omega\\
&=\frac2{m_\sigma^2}
\sup_{u\ge0}u^{2k}\e^{-\sigma^2u^2}.
\end{align*}
For $u>0$, the derivative of the logarithm is
\[
\frac{\d}{\d u}\log\bigl(u^{2k}\e^{-\sigma^2u^2}\bigr)
=\frac{2k}{u}-2\sigma^2u.
\]
It is positive for $u<\sqrt{k}/\sigma$ and negative for
$u>\sqrt{k}/\sigma$, so
\[
\sup_{u\ge0}u^{2k}\e^{-\sigma^2u^2}
=\left(\frac{k}{\sigma^2}\right)^k\e^{-k}.
\]
Taking square roots and using $m_\sigma\ge c_G$, we conclude that
\[
\|\zeta^{(k)}\|_{L^2}
\le\frac{\sqrt2}{m_\sigma}
\left(\frac{\sqrt{k}}{\sigma\sqrt{\e}}\right)^k
\le\frac{\sqrt2}{c_G}
\left(\frac{\sqrt{k}}{\sigma\sqrt{\e}}\right)^k,
\]
which proves \eqref{newkernel-derivatives}.
\end{proof}

\subsection{Finite-time interval recovery}
We now combine the abstract recovery bound with the kernel estimates
to choose the evolution time and smoothing parameter.

\begin{theorem}[Finite-time interval recovery]\label{thm:finite-time-recovery}
For the kernel \eqref{varphik}, let $\bb v$ be the solution of
\eqref{vtpsource}. At $T=5\alpha_{A^{-1}}$,
with $c_G=\operatorname{erf}(1/\sqrt2)$, we have
\begin{equation}\label{interval-continuous-error}
\sup_{|p|\le1}\|\bb v(T,p)-\bb x\|\le\frac{\e^{-r}}{c_G}\|\bb x\|.
\end{equation}
In particular, for any $0<\delta<1$, choosing an integer
\[
r\ge\max\{2,\log(1/(c_G\delta))\}
\]
gives
\[
\sup_{|p|\le1}\|\bb v(T,p)-\bb x\|\le\delta\|\bb x\|.
\]
\end{theorem}
\begin{proof}
Theorem~\ref{thm:newkernel} gives
\[
0\le1-F(p)\le\frac{\e^{-r}}{2c_G}\quad(|p|\le1),\qquad
0\le F(s)\le\frac{\e^{-r}}{2c_G}\quad(|s|\ge4).
\]
Since $T/\alpha_{A^{-1}}-1=4$, substituting these estimates into
\eqref{abstract-interval-error} with $a=1$ proves
\eqref{interval-continuous-error}. The prescribed-accuracy statement follows
from $\e^{-r}/c_G\le\delta$.
\end{proof}

\section{Fourier discretization and interval recovery}\label{sec:discretization}
We discretize \eqref{vtpsource} by Fourier projection after
periodizing its source. We bound the contribution of the periodic
copies on the recovery interval and then control the omitted Fourier
modes. Sampling the projected source gives the finite inhomogeneous
ODE used by the quantum algorithm.

\subsection{Periodization of the convection equations}
The Gaussian-smoothed kernel has rapidly decaying tails. We therefore
periodize it by summing its translates:
\begin{equation}\label{periodic-kernel}
\zeta_R(p)=\sum_{\ell\in\mathbb Z}\zeta(p+2\ell R).
\end{equation}
The series and each differentiated series converge uniformly on
$[-R,R]$, since each derivative of $\zeta$ is a convolution of a
Gaussian derivative with the compactly supported function $f$.
Thus $\zeta_R$ is smooth and $2R$-periodic. We choose
\begin{equation}\label{period-radius}
R=\alpha_AT+4.
\end{equation}
Unlike a compactly supported cut-off, this construction introduces a
small periodization error, which we include explicitly below.

\begin{lemma}\label{lem:truncSchr}
Let $r\ge2$, $\sigma=(2r)^{-1/2}$, and let $R$ be given by
\eqref{period-radius}. Consider
\begin{equation}\label{perExtension}
\begin{cases}
\partial_t\bb v^{\mathrm{per}}(t,p)
=A\partial_p\bb v^{\mathrm{per}}(t,p)+\zeta_R(p)\bb b,
&0<t<T,\quad -R<p<R,\\
\bb v^{\mathrm{per}}(0,p)=\bb0,\\
\bb v^{\mathrm{per}}(t,-R)=\bb v^{\mathrm{per}}(t,R).
\end{cases}
\end{equation}
There is an absolute constant $C$ such that
\begin{equation}\label{periodization-error}
\sup_{\substack{0\le t\le T\\|p|\le1}}
\|\bb v^{\mathrm{per}}(t,p)-\bb v(t,p)\|
\le CT\sigma^{-1}\e^{-4r}\|\bb b\|.
\end{equation}
Moreover,
\begin{equation}\label{periodic-kernel-norm}
\|\zeta_R\|_{L^2(-R,R)}\le\sqrt2/c_G.
\end{equation}
\end{lemma}
\begin{proof}
The convolution formula and $m_\sigma\ge c_G$ imply
\[
|\zeta(q)|\le C\sigma^{-1}
\exp\left(-\frac{(|q|-3)^2}{2\sigma^2}\right),\qquad |q|\ge3.
\]
For $|p|\le1$ and $0\le s\le t\le T$, write
$q=p+\lambda_j(t-s)$ in an eigenbasis of $A$. Then
$|q|\le1+\alpha_AT=R-3$, and, for $\ell\ne0$,
\[
|q+2\ell R|-3\ge(2|\ell|-1)R.
\]
Consequently,
\[
\sum_{\ell\ne0}|\zeta(q+2\ell R)|
\le C\sigma^{-1}\sum_{m=1}^\infty
\e^{-r(2m-1)^2R^2}
\le C\sigma^{-1}\e^{-4r},
\]
where $R\ge4$ makes the remaining Gaussian series uniformly bounded.
The characteristic formula gives
\[
v_j^{\mathrm{per}}(t,p)-v_j(t,p)
=b_j\int_0^t\sum_{\ell\ne0}
\zeta(p+\lambda_j(t-s)+2\ell R)\,\d s.
\]
Taking the Euclidean norm proves \eqref{periodization-error},
for either sign of the eigenvalues.

For the norm bound, let $G_{\sigma,R}$ and $f_R$ be the
periodizations of $G_\sigma$ and $f$. With the unnormalized periodic
convolution $(a*_{\rm per}b)(p)=\int_{-R}^Ra(p-y)b(y)\,\d y$,
we have $\zeta_R=m_\sigma^{-1}G_{\sigma,R}*_{\rm per}f_R$.
Here $\|G_{\sigma,R}\|_1=1$ and $\|f_R\|_2=\sqrt2$, since $R\ge4$.
Young's inequality proves \eqref{periodic-kernel-norm}.
\end{proof}

The Fourier coefficients are explicit. For $\nu_k=\pi k/R$,
$k\in\mathbb Z$, unfolding the periodization gives
\begin{equation}\label{periodic-fourier-coefficients}
c_k=\frac1{2R}\int_{-R}^R\zeta_R(p)\e^{-\i\nu_kp}\,\d p
=-\frac{\i}{Rm_\sigma}\e^{-\sigma^2\nu_k^2/2}
\int_2^3\sin(\nu_ky)\,\d y,\qquad c_0=0.
\end{equation}
For nonzero $\nu_k$, the integral equals
$(\cos(2\nu_k)-\cos(3\nu_k))/\nu_k$.

\begin{remark}[Earlier cut-off parameters]\label{rem:legacy-cutoff}
For comparison, the cut-off construction used in the numerical
experiments below takes
\begin{equation}\label{Rnew}
R_0=\alpha_AT+1,\qquad d=r,\qquad R_1=R_0+2d,\qquad
R=R_1+\alpha_AT=2\alpha_AT+1+2r.
\end{equation}
Its profile is $\psi=\rho\zeta$, where
$\rho=\eta_d*\mathbf1_{[-R_0-d,R_0+d]}$,
$\eta_d(p)=d^{-1}\eta(p/d)$, and $\eta$ is the normalized bump
proportional to $\exp(-1/(1-p^2))$ on $|p|<1$ and zero elsewhere.
Those experiments retain this earlier cut-off discretization and use $\bb\psi=(\psi(p_j))_j$, sampled from the cut-off profile.
The theoretical construction and bounds in the following sections
instead use \eqref{period-radius} and the projected periodic source.
\end{remark}

\subsection{Discretization of the auxiliary variable}	

We apply Fourier projection to \eqref{perExtension} on $[-R,R]$,
using $N_p=2^{n_p}$ equally spaced points, with $n_p\ge2$:
\[
\mathcal I_p=\{0,\ldots,N_p-1\},\qquad
\Delta p=2R/N_p,\qquad p_k=-R+k\Delta p,\quad k\in\mathcal I_p.
\]
Set $p_{N_p}=R$ for the last cell endpoint, identified with $p_0=-R$ by periodicity. Define the accepted grid indices and their number by
\[
\mathcal R_p=\{k\in\mathcal I_p:|p_k|\le1\},\qquad M_p=|\mathcal R_p|.
\]
We label the retained Fourier modes by $l\in\mathcal I_p$ and use
the centred frequencies and shifted Fourier modes
\[
\phi_l(p)=\frac1{\sqrt{N_p}}\e^{\i\mu_l(p+R)},\qquad
\mu_l=\frac{\pi(l-N_p/2)}R,\qquad l\in\mathcal I_p.
\]
For the full Fourier series and its omitted modes, we extend these
same formulas to all $l\in\mathbb Z$.
This scaling makes the retained modes orthonormal on the grid:
\[
\sum_{j\in\mathcal I_p}\overline{\phi_l(p_j)}\phi_m(p_j)=\delta_{lm},
\qquad
\int_{-R}^R\overline{\phi_l(p)}\phi_m(p)\,\d p
=\Delta p\,\delta_{lm}.
\]
The largest retained frequency magnitude is
\[\mu_{\max} = \max_{l\in\mathcal I_p}|\mu_l| = \frac{\pi N_p}{2R} = \frac{\pi}{\Delta p}.\]
The corresponding Fourier space is
\[
S_p=\operatorname{span}\{\phi_l:l\in\mathcal I_p\}.
\]
Let $\mathcal P_p$ be the $L^2(-R,R)$-orthogonal projection onto
$S_p$. For a periodic function $u$,
\begin{equation}\label{uInterpolation}
u_{\Pi}(p):=\mathcal P_pu(p)=\sum_{l\in\mathcal I_p}\widehat u_l\phi_l(p),
\qquad
\widehat u_l=\frac{\sqrt{N_p}}{2R}\int_{-R}^R
u(p)\e^{-\i\mu_l(p+R)}\,\d p.
\end{equation}
Define the projected source profile
\[
\psi_{\Pi}=\mathcal P_p\zeta_R.
\]
Unfolding the periodization gives the coefficients, for all
$l\in\mathbb Z$,
\begin{equation}\label{periodic-fourier-coefficients}
\begin{aligned}
a_l&=\frac{\sqrt{N_p}}{2R}\int_{-R}^R\zeta_R(p)
\e^{-\i\mu_l(p+R)}\,\d p\\
&=-\frac{\i\sqrt{N_p}\,\e^{-\i\mu_lR}}{R m_\sigma}
\e^{-\sigma^2\mu_l^2/2}\int_2^3\sin(\mu_ly)\,\d y.
\end{aligned}
\end{equation}
The phase $\e^{-\i\mu_lR}$ comes from the shifted basis $\phi_l$.
The zero-frequency mode has index $l=N_p/2$, so $a_{N_p/2}=0$.
For $\mu_l\ne0$, the integral equals
$(\cos(2\mu_l)-\cos(3\mu_l))/\mu_l$.
We define the unitary Fourier matrix by
\[
\Phi_{jl}=\phi_l(p_j)
=\frac{(-1)^j}{\sqrt{N_p}}\e^{2\pi\i jl/N_p},
\qquad j,l\in\mathcal I_p.
\]
Discrete orthogonality gives $\Phi^\dagger\Phi=I$.
The matrix $\Phi$ differs from the standard QFT matrix by the
left factor $\diag\bigl((-1)^j\bigr)_{j\in\mathcal I_p}$, which accounts
for the centred frequencies.
For $u_{\Pi}\in S_p$, let $\bb u=(u_{\Pi}(p_j))_j$ and let
$\widehat{\bb u}$ contain its coefficients in the basis
$\{\phi_l\}_{l\in\mathcal I_p}$. Then
\[
\bb u=\Phi\widehat{\bb u},\qquad
\widehat{\bb u}=\Phi^\dagger\bb u.
\]
Differentiation at the nodes is represented by
\[
P_\mu=\Phi D_\mu\Phi^\dagger,\qquad
D_\mu=\diag(\mu_0,\ldots,\mu_{N_p-1}),\qquad
(u_{\Pi}'(p_j))_j=\i P_\mu\bb u.
\]
We define the projected solution by applying $\mathcal P_p$ componentwise:
\[
\bb v_{\Pi}(t,p):=\mathcal P_p\bb v^{\mathrm{per}}(t,p)
\in S_p\otimes\mathbb C^N.
\]
Since $A$ is independent of $p$, the projection commutes with
$A\partial_p$. Projecting \eqref{perExtension} therefore gives
\[
\partial_t\bb v_{\Pi}(t,p)
=A\partial_p\bb v_{\Pi}(t,p)+\psi_{\Pi}(p)\bb b,
\qquad \bb v_{\Pi}(0,p)=\bb0.
\]
We collect the nodal values of the projected solution and source in
\[
\bb U_{\Pi}(t)=[\bb v_{\Pi}(t,p_0);\ldots;\bb v_{\Pi}(t,p_{N_p-1})],
\qquad
\bb\psi_{\Pi}=[\psi_{\Pi}(p_0),\ldots,\psi_{\Pi}(p_{N_p-1})]^\top.
\]
Here the semicolons denote vertical concatenation. Evaluating the
projected equation at the grid points and using the nodal
differentiation matrix $\i P_\mu$, we obtain
\begin{equation}\label{ODEschr}
\dfrac{\d\bb U_{\Pi}(t)}{\d t}=\bar A\bb U_{\Pi}(t)+\bar{\bb b},\qquad
\bb U_{\Pi}(0)=\bb0,\qquad
\bar A=\i(P_\mu\otimes A),\quad
\bar{\bb b}=\bb\psi_{\Pi}\otimes\bb b.
\end{equation}
This is the nodal representation of the Fourier-projected equation,
to which we apply Schr\"odingerization
in Section~\ref{sec:algorithm}.

In discrete Fourier coordinates, with
$\widetilde{\bb U}_{\Pi}=(\Phi^\dagger\otimes I)\bb U_{\Pi}$ and
$\widetilde{\bb\psi}_{\Pi}=\Phi^\dagger\bb\psi_{\Pi}$, the same system reads
\[
\dfrac{\d\widetilde{\bb U}_{\Pi}(t)}{\d t}
=\i(D_\mu\otimes A)\widetilde{\bb U}_{\Pi}(t)
+\widetilde{\bb\psi}_{\Pi}\otimes\bb b,
\qquad \widetilde{\bb U}_{\Pi}(0)=\bb0.
\]
We solve this equivalent system in Fourier coordinates. Its source
vector is $\widetilde{\bb\psi}_{\Pi}=(a_l)_{l\in\mathcal I_p}$, with explicit entries
\begin{equation}\label{fourier-source-vector}
(\widetilde{\bb\psi}_{\Pi})_l=
\begin{cases}
-\dfrac{\i\sqrt{N_p}}{Rm_\sigma}\,
\e^{-\i\mu_lR-\sigma^2\mu_l^2/2}
\dfrac{\cos(2\mu_l)-\cos(3\mu_l)}{\mu_l},
& l\ne N_p/2,\\[6pt]
0,&l=N_p/2.
\end{cases}
\end{equation}
With $m_\sigma$ fixed, these entries can be evaluated directly.
We prepare $\widetilde{\bb\psi}_{\Pi}/\|\widetilde{\bb\psi}_{\Pi}\|$ in the
Fourier register and apply $\Phi$ to recover nodal coordinates after
evolution. This transform is implemented by a QFT followed by the
diagonal phase correction described above. Since $\Phi$ is unitary,
$\|\widetilde{\bb\psi}_{\Pi}\|=\|\bb\psi_{\Pi}\|$.

\subsection{Recovery and discretization error}
For any chosen grid index $k\in\mathcal R_p$, we define the recovered approximation by
\begin{equation}\label{xhT}
\bb{x}_{\Pi,T}=\Pi_k\bb U_{\Pi}(T)=\bb v_{\Pi}(T,p_k),
\qquad \Pi_k=\bra k\otimes I.
\end{equation}
Each coordinate-extraction map $\Pi_k$ has norm one.
We suppress the dependence on $k$ in $\bb{x}_{\Pi,T}$; all recovery
error bounds below hold uniformly over $k\in\mathcal R_p$.
We combine these nodal approximations through the coherent interval
projection defined below.
The projection error is controlled directly by the Gaussian Fourier tail.

\begin{lemma}\label{lem:L2errspectral}
Let $N_p\ge4$, and set $K=\mu_{N_p-1}=\pi(N_p/2-1)/R$.
If $K\ge1$, then
\begin{equation}\label{gaussian-fourier-tail}
\frac1{\sqrt{N_p}}
\sum_{l\in\mathbb Z\setminus\mathcal I_p}|a_l|
\le C\frac{\e^{-\sigma^2K^2/2}}{\sigma^2K^2}.
\end{equation}
In particular, the same bound holds for
$\|\zeta_R-\mathcal P_p\zeta_R\|_{L^\infty(-R,R)}$.
\end{lemma}
\begin{proof}
Equation~\eqref{periodic-fourier-coefficients} gives
$|a_l|\le C\sqrt{N_p}(R|\mu_l|)^{-1}\e^{-\sigma^2\mu_l^2/2}$
for $l\ne N_p/2$. The omitted indices are $l<0$ and $l\ge N_p$.
With $h=\pi/R$, their frequency magnitudes are
$(N_p/2+1+m)h$ and $(N_p/2+m)h$, respectively, for $m\ge0$.
Monotonicity of $u^{-1}\e^{-\sigma^2u^2/2}$ therefore gives
\[
\begin{aligned}
\frac1{\sqrt{N_p}}
\sum_{l\in\mathbb Z\setminus\mathcal I_p}|a_l|
&\le\frac C R\sum_{m=0}^\infty
\frac{\e^{-\sigma^2[(N_p/2+m)h]^2/2}}{(N_p/2+m)h}\\
&\le C\int_K^\infty \frac{\e^{-\sigma^2u^2/2}}u\,\d u
\le C\frac{\e^{-\sigma^2K^2/2}}{\sigma^2K^2}.
\end{aligned}
\]
The maximum-norm bound follows by summing the omitted modes and using
$|\phi_l(p)|=N_p^{-1/2}$.
\end{proof}

We now combine continuous recovery, periodization, and Fourier projection.
\begin{theorem}\label{thm:pdiscretization}
Let the periodic source and domain be as in
Lemma~\ref{lem:truncSchr}. For every $0<\varepsilon\le1/2$, there are
choices of the smoothing and grid parameters for which
\begin{equation}\label{error2}
\|\bb x-\bb{x}_{\Pi,T}\|
\le \varepsilon\|\bb x\|.
\end{equation}
These choices can be made with
\[
r=\Theta(\log(2\kappa_A/(\xi\alpha_A\varepsilon))),\qquad
\mu_{\max}=\Theta(r),\qquad N_p=\Theta(\kappa_A r).
\]
Enlarging the absolute constants in
the parameter choices permits any fixed fraction of $\varepsilon$ in
\eqref{error2}.
\end{theorem}

\begin{proof}
Introduce the periodic homogeneous response
\[
\partial_t\bb w=A\partial_p\bb w,\qquad
\bb w(0,p)=\zeta_R(p)\bb b.
\]
Its projection evolves independently in each retained mode:
\[
\bb w_{\Pi}(t,p)=\mathcal P_p\bb w(t,p)
=\sum_{l\in\mathcal I_p}a_l\phi_l(p)\e^{\i\mu_lAt}\bb b.
\]
Since the projection commutes with $A\partial_p$, Duhamel's formula
gives
\[
\bb v^{\mathrm{per}}(T,p)=\int_0^T\bb w(t,p)\,\d t,\qquad
\bb v_{\Pi}(T,p)=\int_0^T\bb w_{\Pi}(t,p)\,\d t.
\]
Each $\e^{\i\mu_lAt}$ is unitary. Thus
Lemma~\ref{lem:L2errspectral} gives, uniformly in $t$ and $p$,
\begin{equation}\label{errorpk}
\|\bb w(t,p)-\bb w_{\Pi}(t,p)\|
\le C\frac{\e^{-\sigma^2K^2/2}}{\sigma^2K^2}\|\bb b\|.
\end{equation}
There is no aliasing term, because the source coefficients are
projected before sampling.

For every accepted point $p_k$, the error decomposes as
\[
\begin{aligned}
\bb x-\Pi_k\bb U_{\Pi}(T)
={}&\bb x-\bb v(T,p_k)
+\bb v(T,p_k)-\bb v^{\mathrm{per}}(T,p_k)\\
&+\int_0^T[\bb w(t,p_k)-\bb w_{\Pi}(t,p_k)]\,\d t.
\end{aligned}
\]
Theorem~\ref{thm:finite-time-recovery},
\eqref{periodization-error}, and \eqref{errorpk} therefore imply
\[
\max_{k\in\mathcal R_p}
\frac{\|\bb x-\Pi_k\bb U_{\Pi}(T)\|}{\|\bb x\|}
\le \frac{\e^{-r}}{c_G}
+C\frac T\xi\left(
\sqrt r\,\e^{-4r}
+\frac{\e^{-\sigma^2K^2/2}}{\sigma^2K^2}\right).
\]
Choose $N_p$ as the smallest power of two with
$N_p\ge\max\{4,2RC_1r/\pi\}$, where $C_1\ge8$ is an absolute
constant. Then $\mu_{\max}=\Theta(r)$ and
$K\ge\mu_{\max}/2\ge C_1r/2$, so
$\sigma^2K^2/2\ge C_1^2r/16\ge4r$.
Since $T=5\alpha_{A^{-1}}$ and $\sqrt r\le\e^r$, the relative
error is bounded by
\[
\frac{\e^{-r}}{c_G}
+C\frac{\kappa_A}{\xi\alpha_A}\e^{-3r}.
\]
Taking the least integer
\[
r\ge\max\left\{r_0,\,
C_0\log\frac{2\kappa_A}{\xi\alpha_A\varepsilon}\right\}
\]
with sufficiently large absolute constants $r_0\ge2$ and $C_0$
makes this bound at most $\varepsilon$.
Here $\kappa_A/(\xi\alpha_A)\ge1$.
Finally, \eqref{period-radius} gives $R=5\kappa_A+4=\Theta(\kappa_A)$,
and hence $N_p=\Theta(\kappa_A r)$.
Increasing the constants permits any fixed fraction of $\varepsilon$.
\end{proof}
\begin{remark}\label{rem:muMax}
The Gaussian Fourier tail determines the required frequency scale:
$\mu_{\max}=\Theta(r)$ suffices because $\sigma^2=1/(2r)$.
The parameter $r$ can be selected from a certified lower bound on $\xi$. If a certified constant-factor estimate of $\xi$ is available, it permits
$r=\Theta(\log(2\kappa_A/(\xi\alpha_A\varepsilon)))$.
Without such an estimate, the known bound $\xi\ge1/\alpha_A$ permits
$r=\Theta(\log(2\kappa_A/\varepsilon))$.
All subsequent query bounds use the actual chosen value of $r$.
\end{remark}

\subsection{Interval recovery probability}
We now quantify the contribution of interval recovery to the success
probability. The uniform error bound in \eqref{error2} gives, for
$0<\varepsilon\le1/2$,
\[
(1-\varepsilon)^2M_p\|\bb x\|^2
\le\sum_{k\in\mathcal R_p}\|\Pi_k\bb U_{\Pi}(T)\|^2
\le(1+\varepsilon)^2M_p\|\bb x\|^2,
\]
where $M_p$ is the number of grid points in the recovery interval $[-1,1]$.
The squared norms add because distinct grid labels are orthogonal.
For pure-state recovery, we project coherently onto
$\ket{\eta_p}=M_p^{-1/2}\sum_{k\in\mathcal R_p}\ket k$.
The same uniform error bound gives
\[
\left\|M_p^{-1/2}\sum_{k\in\mathcal R_p}\Pi_k\bb U_{\Pi}(T)
-\sqrt{M_p}\,\bb x\right\|
\le\varepsilon\sqrt{M_p}\,\|\bb x\|.
\]
Hence this projection has the same probability scaling as accepting
all grid points in the interval.
As established in the proof of Theorem~\ref{thm:solver-queries},
the quantum implementation constructed below has single-attempt success
probability $\mathbb{P}_p$ before amplitude amplification, with
\[
\mathbb{P}_p
=\Theta\!\left(
\frac{\sum_{k\in\mathcal R_p}\|\Pi_k\bb U_{\Pi}(T)\|^2}
{T^2\|\bb\psi_{\Pi}\|^2\|\bb b\|^2}\right)
=\Theta\!\left(\frac{M_p}{\|\bb\psi_{\Pi}\|^2}\frac{\xi^2}{T^2}\right),
\]
This relation includes the
second-variable recovery and the implementation error bounds proved
there. Here we establish the first-variable factor
$M_p/\|\bb\psi_{\Pi}\|^2=\Theta(1)$, which ensures that grid refinement
introduces no additional loss in recovery probability.

We first establish a positive lower bound for the projected source.
The original kernel satisfies
\[
\int_1^4\zeta(p)\,\d p=F(1)-F(4)
\ge1-c_G^{-1}\e^{-r}\ge1/2,\qquad r\ge2.
\]
Since $T=5\alpha_{A^{-1}}$ and $\kappa_A\ge1$, the new domain has
$R=5\kappa_A+4\ge9$. For $p\in[1,4]$ the nonzero periodic copies
satisfy $|p+2\ell R|-3\ge2|\ell|R-7$.
The Gaussian tail estimate used in Lemma~\ref{lem:truncSchr}
therefore implies
\[
\sup_{p\in[1,4]}|\zeta_R(p)-\zeta(p)|
\le C\sqrt r\,\e^{-4r}.
\]
Together with \eqref{gaussian-fourier-tail} and the chosen frequency
scale, this yields
$\sup_{p\in[1,4]}|\psi_{\Pi}(p)-\zeta(p)|
\le C\sqrt r\,\e^{-4r}$.
Enlarging $r_0$ ensures
$|\int_1^4\psi_{\Pi}(p)\,\d p|\ge1/4$. Hence
\[
\|\psi_{\Pi}\|_{L^2(-R,R)}^2
\ge\frac13\left|\int_1^4\psi_{\Pi}(p)\,\d p\right|^2
\ge\frac1{48}.
\]
The upper bound follows from the contractivity of the orthogonal
projection and \eqref{periodic-kernel-norm}.
Finally, discrete orthogonality of the retained modes gives the exact
identity
\[
\Delta p\,\|\bb\psi_{\Pi}\|^2
=\Delta p\sum_{l\in\mathcal I_p}|a_l|^2
=\|\psi_{\Pi}\|_{L^2(-R,R)}^2.
\]
Thus no quadrature estimate is needed, and $\Delta p=\Theta(1/r)$ gives
\begin{equation}\label{norm_psi}
c\le \Delta p\,\|\bb\psi_{\Pi}\|^2\le C,
\qquad \|\bb\psi_{\Pi}\|=\Theta(\sqrt r).
\end{equation}
Because $N_p$ is even and the grid is symmetric,
$p_{N_p/2}=0$. The accepted points are therefore the integer
multiples of $\Delta p$ with modulus at most one, so
\[
M_p=2\left\lfloor\frac1{\Delta p}\right\rfloor+1,
\qquad
2-\Delta p\le M_p\Delta p\le2+\Delta p.
\]
To keep these grid-count bounds uniform, it suffices that
$\Delta p\le1/2$. This follows from $\Delta p=\Theta(1/r)$
after a further constant enlargement of $r_0$.
Together with \eqref{norm_psi}, this implies
\begin{equation}\label{interval-grid-mass}
1\le M_p\Delta p\le3,\qquad
M_p=\Theta(r),\qquad
c\le\frac{M_p}{\|\bb\psi_{\Pi}\|^2}\le C.
\end{equation}
Substituting this ratio into the probability expression above gives
\[
\mathbb{P}_p=\Theta\!\left(\frac{\xi^2}{T^2}\right)
=\Theta\!\left(\frac{\xi^2\alpha_A^2}{\kappa_A^2}\right).
\]
\begin{remark}[Kernel choice and recovery in the two auxiliary variables]
We solve the finite-dimensional ODE obtained here using the
Schr\"odingerization method of \cite{DMPYSchrodingerizationODE}, as
described in Section~\ref{sec:algorithm}. In that method, the auxiliary
initial profile in $q$ approximates $\e^{-q}$ on the region relevant
to recovery. The ODE solution is recovered by a weighted projection
with $\int\chi_q(q)\e^{-q}\,\d q=1$. The profile and recovery weight
are thus chosen to reproduce and remove this known exponential factor.

The source kernel $\zeta$ in our first auxiliary variable $p$ serves
a different purpose: its tail integral $F(p)$ must approximate one
on a fixed interval, so that $\bb v(T,p)$ approximates $\bb x$
uniformly there. We then recover through a uniform superposition of
the accepted nodes. Although $\|\bb\psi_{\Pi}\|^2=\Theta((\Delta p)^{-1})$,
interval recovery gives $M_p/\|\bb\psi_{\Pi}\|^2=\Theta(1)$.
A bounded continuous $L^2$ norm alone is insufficient: the Gaussian
choice $\zeta(p)=p\e^{-p^2/2}$ already has
$\|\zeta\|_{L^2(\mathbb R)}^2=\sqrt\pi/2$, but
$F(p)=\e^{-p^2/2}$ does not approach one uniformly on a fixed interval;
see Remark~\ref{rem:interval-accuracy}.
\end{remark}

\section{Schr\"odingerization and query bounds}\label{sec:algorithm}
We implement the two-variable Schr\"odingerization representation
using Hamiltonian simulation. The kernel and recovery
estimates preserve a single logarithmic precision factor; block
preconditioning then gives linear scaling in $\kappa_A$ without VTAA.

\subsection{Schr\"odingerization and query complexity}\label{subsec:schrodingerization}
We apply the Schr\"odingerization method
\cite{JLY22SchrShort,JLY22SchrLong,DMPYSchrodingerizationODE}
to \eqref{ODEschr} by introducing a second auxiliary variable $q$.
Together with the convection variable $p$, this realizes
Schr\"odingerization of the QLSP in two higher dimensions.
We use the following result
from \cite[Theorem~4.1]{DMPYSchrodingerizationODE}.

\begin{lemma}[Autonomous dissipative evolution]\label{lem:autonomous-dissipative}
Consider
\[
\frac{\d\bb w(t)}{\d t}=(B_1+\i B_2)\bb w(t),\qquad
B_1=B_1^\dagger\preceq 0,\quad B_2=B_2^\dagger,
\]
with $\bb w(0)\ne\bb0$, $\bb w(T)\ne\bb0$, and
$g_w=\|\bb w(0)\|/\|\bb w(T)\|$.
Assume exact block encodings of $B_1,B_2$ with positive normalizations
$\alpha_1,\alpha_2$, and exact preparation of the normalized initial
state and the auxiliary profile states, including the required controls
and inverses. For $0<\varepsilon_{\rm ODE}\le1/4$, there exists a quantum algorithm
that prepares an $\varepsilon_{\rm ODE}$-approximation in Euclidean norm to
$\bb w(T)/\|\bb w(T)\|$, with a success flag of probability $\Theta(1)$,
using
\[
\mathcal O\!\left(
g_w\left[(\alpha_1T+1)\log\frac{g_w}{\varepsilon_{\rm ODE}}+\alpha_2T\right]
\right)
\]
queries to the matrix block encodings and $\mathcal O(g_w)$ queries
to each state-preparation oracle.
\end{lemma}

For our constant inhomogeneous term, the construction in
\cite[Section~6.2]{DMPYSchrodingerizationODE} gives
\begin{equation}\label{utsourceFourierAugment}
\frac{\d\bb U_{\Pi,f}(t)}{\d t}=\bar A_f\bb U_{\Pi,f}(t),\qquad
\bar A_f=\begin{bmatrix}\bar A&I/T\\0&0\end{bmatrix},\qquad
\bb U_{\Pi,f}(t)=\begin{bmatrix}\bb U_{\Pi}(t)\\T\bar{\bb b}\end{bmatrix},
\quad
\bb U_{\Pi,f}(0)=\begin{bmatrix}\bb0\\T\bar{\bb b}\end{bmatrix}.
\end{equation}
Since $\bar A^\dagger=-\bar A$, we have $\bar A_f=H_1+\i H_2$, where
\begin{equation}\label{H1H2def}
H_1=\frac1{2T}\begin{bmatrix}0&I\\I&0\end{bmatrix},\qquad
H_2=\begin{bmatrix}P_\mu\otimes A&-\i I/(2T)\\
\i I/(2T)&0\end{bmatrix}.
\end{equation}
The scalar shift used in that reference makes the generator dissipative:
\[
\widehat A_f=\bar A_f-\frac{I}{2T},\qquad
\frac{\widehat A_f+\widehat A_f^\dagger}{2}
=-\frac1{2T}\begin{bmatrix}I&-I\\-I&I\end{bmatrix}\preceq 0.
\]
Thus $\bb V(t)=\e^{-t/(2T)}\bb U_{\Pi,f}(t)$ satisfies
$\d\bb V/\d t=\widehat A_f\bb V$ and
$\bb V(T)=\e^{-1/2}\bb U_{\Pi,f}(T)$.
We apply Lemma~\ref{lem:autonomous-dissipative} to prepare the normalized
state of $\bb U_{\Pi,f}(T)$, then select its upper block and project the
first-variable register onto $\ket{\eta_p}$. This projection is implemented
by inverse preparation of $\ket{\eta_p}$ followed by selection of the
zero state. We amplify the joint success flag to obtain the
QLSP solution with constant probability.

We count queries to an exact matrix block encoding $U_A$ of $A$ with
normalization $\alpha_A$ and to $O_b\ket0=\ket b$, including controls
and inverses. We assume the auxiliary state-preparation access required
by Lemma~\ref{lem:autonomous-dissipative}, including preparation of the normalized source $\widetilde{\bb\psi}_{\Pi}/\|\widetilde{\bb\psi}_{\Pi}\|$ in Fourier coordinates. The explicit coefficients in \eqref{fourier-source-vector} specify this state but do not by themselves establish its gate complexity; auxiliary preparation costs are not included in the input-oracle query count. These profiles are independent
of the input oracles. The resulting bounds are summarized below.

\begin{theorem}[QLSP preparation and query complexity]\label{thm:solver-queries}
Let $A$ be invertible and suppose a valid constant-factor estimate of
$\xi$ is available. Under the preceding access assumptions,
for $0<\varepsilon\le1/2$ there exists a quantum algorithm that prepares
an $\varepsilon$-approximation to $\ket x$ in $\ell^2$ norm, with a
success flag of probability $\Omega(1)$. The algorithm uses
\begin{equation}\label{log575}
Q_A=\mathcal O\!\left(\frac{\kappa_A^2}{\xi\alpha_A}
\log\frac{\kappa_A}{\xi\alpha_A\varepsilon}\right),
\qquad Q_b=\mathcal O\!\left(\frac{\kappa_A}{\xi\alpha_A}\right)
\end{equation}
queries to $U_A$ and $O_b$, respectively.
\end{theorem}
\begin{proof}
First suppose $A$ is Hermitian. Since $\bar A$ is skew-Hermitian,
\[
\bb U_{\Pi}(T)=\int_0^T\e^{(T-s)\bar A}\bar{\bb b}\,\d s,
\qquad \|\bb U_{\Pi}(T)\|\le T\|\bar{\bb b}\|.
\]
Consequently,
\[
T^2\|\bar{\bb b}\|^2\le\|\bb U_{\Pi,f}(T)\|^2
\le2T^2\|\bar{\bb b}\|^2,\qquad
\sqrt{\e/2}\le
\frac{\|\bb V(0)\|}{\|\bb V(T)\|}\le\sqrt\e.
\]
Thus the norm ratio in Lemma~\ref{lem:autonomous-dissipative} is bounded
by an absolute constant. For the shifted generator, standard block-encoding
constructions give
\[
B_1=H_1-I/(2T),\quad B_2=H_2,\qquad
\alpha_1=1/T,\quad \alpha_2=\mu_{\max}\alpha_A+1/(2T).
\]
Each matrix query uses $\mathcal O(1)$ calls to $U_A$ and its inverse;
preparing $\bb U_{\Pi,f}(0)/\|\bb U_{\Pi,f}(0)\|=\ket1\otimes
(\bb\psi_{\Pi}/\|\bb\psi_{\Pi}\|)\otimes\ket b$ uses one call to $O_b$.
The lemma therefore prepares the augmented state to accuracy $\varepsilon_{\rm ODE}$
with $\mathcal O(\mu_{\max}\alpha_AT+\log(1/\varepsilon_{\rm ODE}))$ matrix queries
and $\mathcal O(1)$ right-hand-side queries.

Let $g$ be the norm of the accepted component of the exact normalized
augmented state. The interval error and normalization estimates give
\[
g^2=\frac{\left\|M_p^{-1/2}\sum_{k\in\mathcal R_p}\Pi_k\bb U_{\Pi}(T)\right\|^2}
{\|\bb U_{\Pi,f}(T)\|^2}
=\Theta\!\left(\frac{M_p\xi^2}{T^2\|\bb\psi_{\Pi}\|^2}\right)
=\Theta\!\left(\frac{\xi^2\alpha_A^2}{\kappa_A^2}\right).
\]
The auxiliary ODE solver has a success flag of constant probability,
so including this flag preserves the order of $g^2$.
Choose $\varepsilon_{\rm ODE}=c\varepsilon\xi/T$ with a sufficiently small
absolute constant $c>0$, using the supplied estimate of $\xi$.
Projection is a contraction, and normalization increases the vector
error by at most a constant times $1/g$. The ODE-solver error therefore
contributes $\mathcal O(\varepsilon)$ after recovery.
Theorem~\ref{thm:pdiscretization} and the coherent interval-recovery
bound give the same order of error in the normalized recovered vector
relative to $\ket x$. Allocating sufficiently small fixed fractions
of $\varepsilon$ to these two errors gives a pure output state
$\ket{\widetilde x}$ satisfying
\[
\|\ket{\widetilde x}-\ket x\|\le\varepsilon.
\]

We choose $r=\Theta(\log(\kappa_A/(\xi\alpha_A\varepsilon)))$
with a sufficiently large implied constant. Since $\mu_{\max}=\Theta(r)$
and $\log(1/\varepsilon_{\rm ODE})=\mathcal O(r)$,
each attempt uses $\mathcal O(\kappa_A r)$ matrix queries and
$\mathcal O(1)$ right-hand-side queries. A known constant-factor lower
bound on $g$ permits bounded amplitude amplification using
$\mathcal O(1/g)=\mathcal O(\kappa_A/(\xi\alpha_A))$ attempts and their
inverses, proving \eqref{log575}.

For general invertible $A$, we use the Hermitian dilation
\[
\widetilde A=\begin{bmatrix}0&A\\A^\dagger&0\end{bmatrix},\qquad
\ket{\widetilde b}=\ket0\otimes\ket b,\qquad
\widetilde A^{-1}\ket{\widetilde b}=\ket1\otimes A^{-1}\ket b.
\]
Its norm bounds and $\xi$ are unchanged, and its block encoding uses
$\mathcal O(1)$ queries to $U_A$ and $U_A^\dagger$.
We also project the extra qubit onto $\ket1$, its value in the target
state. This projection has constant success probability and changes the
normalized-vector error by at most a constant factor, which we absorb
in the error budget.
\end{proof}
\subsection{Linear scaling via block preconditioning}\label{subsec:block-preconditioning}
We apply the block preconditioner of \cite{Low2026quantumlinearsystem}
to obtain linear dependence on $\kappa_A$ while retaining ordinary
amplitude amplification. Thus we do not require VTAA. Let
$S=I-(1-s)\ket b\bra b$, where $0<s\le1$, and consider
\begin{equation}\label{SASb}
SA\bb x=S\bb b.
\end{equation}
Since $\ket{Sb}=\ket b$, we use the same preparation oracle $O_b$.
The construction in \cite[Eq.~(176)]{Low2026quantumlinearsystem}
block-encodes $S$ with normalization one using two queries to $O_b$
and its inverse. Combining it with $U_A$ gives an $\alpha_A$-normalized
block encoding of $SA$ with constant query overhead.

We recall the properties needed for the complexity bound;
see \cite{Low2026quantumlinearsystem} for the preconditioning analysis.
\begin{lemma}\label{lem:Sproperty}
Let $A$ be invertible and suppose $\xi/c_\xi\le\xi_c\le c_\xi\xi$ for a
fixed factor $c_\xi>1$. Choose
\begin{equation}\label{sChoose}
s=\frac{\xi_c\alpha_A}{c_\xi\kappa_A},\qquad
\frac{\xi\alpha_A}{c_\xi^2\kappa_A}\le s
\le\frac{\xi\alpha_A}{\kappa_A}\le1.
\end{equation}
Then
\begin{align}
\xi_{SA}:=\|(SA)^{-1}\ket{Sb}\|
&=\frac{\xi}{s},\qquad
\alpha_{A^{-1}}\le\xi_{SA}\le c_\xi^2\alpha_{A^{-1}},\label{property1}\\
\|SA\|&\le\|A\|\le\alpha_A,\label{property2}\\
\|(SA)^{-1}\|&\le\sqrt{1+c_\xi^4}\,\alpha_{A^{-1}}.\label{property3}
\end{align}
The normalized solution of \eqref{SASb} is $\ket x$.
\end{lemma}

Combining these properties with Theorem~\ref{thm:solver-queries},
we obtain the following result.
\begin{theorem}[Linear query complexity]\label{thm:optimal-queries}
Under the access assumptions of Theorem~\ref{thm:solver-queries}
and the conditions of Lemma~\ref{lem:Sproperty}, for
$0<\varepsilon\le1/2$ there exists a quantum algorithm that prepares
an $\varepsilon$-approximation to $\ket x$ in $\ell^2$ norm, with a
success flag of probability $\Omega(1)$, using
\begin{equation}\label{oplog575}
\mathcal O\!\left(\kappa_A\log\frac1\varepsilon\right)
\end{equation}
queries to each original oracle $U_A$ and $O_b$.
\end{theorem}
\begin{proof}
By Lemma~\ref{lem:Sproperty}, we may take
\[
\alpha_{SA}=\alpha_A,\qquad
\alpha_{(SA)^{-1}}=\sqrt{1+c_\xi^4}\,\alpha_{A^{-1}},\qquad
\kappa_{SA}=\sqrt{1+c_\xi^4}\,\kappa_A.
\]
The bounds on $\xi_{SA}$ also provide the constant-factor estimate
required by Theorem~\ref{thm:solver-queries}, which applies to $SA$
whether or not it is Hermitian. Since $\xi_{SA}=\xi/s$, the choice
\eqref{sChoose} gives
\[
\frac{\kappa_{SA}}{\xi_{SA}\alpha_{SA}}
=\sqrt{1+c_\xi^4}\,\frac{s\kappa_A}{\xi\alpha_A}
=\Theta(1).
\]
Substitution into \eqref{log575} yields
\[
Q_{SA}=\mathcal O\!\left(
\frac{\kappa_{SA}^2}{\xi_{SA}\alpha_{SA}}
\log\frac{\kappa_{SA}}{\xi_{SA}\alpha_{SA}\varepsilon}\right)
=\mathcal O\!\left(\kappa_A\log\frac1\varepsilon\right),
\qquad Q_{Sb}=\mathcal O(1).
\]
Each query to the working matrix uses a constant number of queries
to $U_A$ and $O_b$, while $O_{Sb}=O_b$. This proves the stated
complexity for both original input oracles.
\end{proof}

\begin{remark}
Theorem~\ref{thm:optimal-queries} assumes a valid constant-factor
estimate $\xi_c$. Solution-norm estimation is discussed in
\cite{Costa2022QLSA,dalzell2026shortcutoptimalquantumlinear}; when this
estimate is not supplied, its query cost and confidence must be
accounted for separately.
\end{remark}

\section{Numerical methods and additional results}\label{sec:numerics}
We use UnitaryLab~\cite{UnitaryLabSoftware} to test the periodic initialization and weighted
recovery of \cite{DMPYSchrodingerizationODE} for our QLSP construction.
We examine the recovered solutions, discretization errors, and
the probability gain from interval recovery.

\paragraph{Test systems and discretization.}
We consider a positive-definite system (SPD2) and an indefinite system
with complex data (Asym2):
\[
A_+=\begin{bmatrix}3/2&-1/2\\-1/2&3/2\end{bmatrix},\quad
\bb b_+=\begin{bmatrix}1\\0\end{bmatrix},\qquad
A_{\rm a}=\begin{bmatrix}1&2\\2&1\end{bmatrix},\quad
\bb b_{\rm a}=\frac{2}{\sqrt5}
\begin{bmatrix}1\\(3+4\i)/10\end{bmatrix}.
\]
Their eigenvalues are $\{1,2\}$ and $\{-1,3\}$, respectively.
We also test a four-dimensional discrete Poisson system (Poisson4),
with $A_P=K/\lambda_{\min}(K)$ and
$K=\operatorname{tridiag}(-1,2,-1)\in\mathbb R^{4\times4}$.
We choose its right-hand side proportional to
$[\sin(\pi j/5)+\tfrac12\sin(2\pi j/5)+\tfrac14\sin(3\pi j/5)]_{j=1}^4$,
so several eigenmodes contribute to the solution. We normalize every
right-hand side to $\|\bb b\|=1$.
All three systems have $\|A^{-1}\|=1$, and we take $T=5$.
Unless varied explicitly, we use $r=8$, $R=\alpha_AT+4$ with
$\alpha_A=\|A\|$, and the values of $N_p$ in
Table~\ref{tab:qlsp-dissipative}.
We initialize the $p$-register directly in
$\widetilde{\bb\psi}_{\Pi}/\|\widetilde{\bb\psi}_{\Pi}\|$ using
\eqref{fourier-source-vector} and evolve the ODE in Fourier coordinates.
After applying $\Phi$, we use the nodes indexed by $\mathcal R_p$
for interval recovery.

For the Schr\"odingerization variable $q$, we apply the autonomous
solver to the shifted generator $\widehat A_f$ defined above. Following \cite[Section~4]{DMPYSchrodingerizationODE},
we use the periodic profile and recovery weight
\[
\Psi_q(q)=\sum_{j\in\mathbb Z}\psi_q(q+j\ell_q),\qquad
\psi_q(q)=\frac{\e^{-q}}2
\bigl(1+\operatorname{erf}(2\sqrt L\,q)\bigr),
\]
\[
\chi_q(q)=\frac{2\e^{-q}}{\e^{-1}-\e^{-2}}
\mathbf1_{[1/2,1]}(q),
\]
on $[-L,\ell_q-L)$, where $\ell_q=2+2L$.
We set $L=6$ and $N_q=128$. Let $\mathcal P_q$ be the orthogonal
projection onto the modes
\[
\phi_l^{(q)}(q)=N_q^{-1/2}\e^{\i\nu_l(q+L)},\qquad
\nu_l=\frac{2\pi(l-N_q/2)}{\ell_q},\qquad 0\le l<N_q.
\]
We denote the projected profiles by
\[
\Psi_{q,\Pi}=\mathcal P_q\Psi_q,\qquad
\chi_{q,\Pi}=\mathcal P_q\chi_q,
\]
and prepare the Fourier coefficients of $\Psi_{q,\Pi}$ directly
from their analytic expression.
With $\Delta q=\ell_q/N_q$, the continuous $L^2$ norm of each
projected $q$-profile equals $\sqrt{\Delta q}$ times the Euclidean
norm of its coefficient vector. We include this factor in the
reconstruction below. After evolution, we use the centred Fourier
transform and inverse state preparation to project
onto the normalized $\chi_{q,\Pi}$, select the upper augmentation block,
and project the $p$-register onto $\ket{\eta_p}$.
As for $\Phi$, each centred Fourier transform is implemented by a
QFT followed by the phase $(-1)^j$. This gives a pure system state.

\paragraph{Circuit implementation and error measures.}
We run UnitaryLab 1.1.6 with its NumPy backend in double precision.
In the two Fourier registers, the Hamiltonian is
\[
\mathcal H=\mathcal K+\mathcal J+\mathcal C,\qquad
\mathcal K=-I_{N_q}\otimes D_\mu\otimes\ket0\bra0\otimes A,
\]
\[
\mathcal J=D_\nu\otimes I_{N_p}\otimes\frac{X-I}{2T}\otimes I,
\qquad
\mathcal C=-I_{N_q}\otimes I_{N_p}\otimes\frac{Y}{2T}\otimes I,
\]
where $X,Y$ are Pauli matrices on the augmentation qubit and
$D_\nu=\diag(\nu_0,\ldots,\nu_{N_q-1})$.
For $m$ steps with $\Delta t=T/m$, we implement
\[
\e^{-\i\mathcal HT}\approx
\left(
\e^{-\i\mathcal K\Delta t/2}\e^{-\i\mathcal J\Delta t/2}
\e^{-\i\mathcal C\Delta t}
\e^{-\i\mathcal J\Delta t/2}\e^{-\i\mathcal K\Delta t/2}
\right)^m.
\]
Binary frequency encoding reduces these factors to controlled
small-system unitaries, phase gates, and single-qubit rotations;
for these explicit matrices, the controlled exponentials of $A$
are supplied as small unitary gates.
The circuit includes both auxiliary state preparations and their
recovery operations. As in the numerical experiment of
\cite{DMPYSchrodingerizationODE}, this product-formula implementation
tests the construction and recovery; it does not test the
Hamiltonian-simulation query bound.

Let $\bb y$ be the unnormalized successful branch of the circuit.
We restore the known preparation and recovery factors to obtain
\[
\bb x^{\,a}=C_{\rm rec}\bb y,\qquad
C_{\rm rec}=\frac{\e^{1/2}T\|\bb\psi_{\Pi}\|
\|\Psi_{q,\Pi}\|_{L^2}\|\chi_{q,\Pi}\|_{L^2}}{\sqrt{M_p}}.
\]
The factor $\e^{1/2}$ reverses the dissipative shift; no reference
solution norm is used in this reconstruction. Against the direct
linear solve, we report
\[
E_{\rm vec}=\frac{\|\bb x^{\,a}-\bb x\|}{\|\bb x\|},\qquad
E_{\rm state}=\left\|
\frac{\bb x^{\,a}}{\|\bb x^{\,a}\|}
-\frac{\bb x}{\|\bb x\|}\right\|,\qquad
\mathbb P_p=\|\bb y\|^2.
\]
We retain complex phases when computing both errors.
Independent matrix exponentiation of the Fourier-mode Hamiltonians
provides a reference for separating product-formula error from
profile and spatial errors. Direct quadrature verifies the analytic
source coefficients and their normalization to below $10^{-12}$.
Random-state checks of the circuit
steps and recovery agree with independent matrix calculations
to below $10^{-12}$; the largest measured state-preparation discrepancy
is $4.41\times10^{-9}$.

\paragraph{Accuracy and refinement.}
Table~\ref{tab:qlsp-dissipative} reports the results at $m=512$.
The examples test positive and negative eigenvalues, unequal spectral
weights, and multiple physical modes.
In Fig.~\ref{fig:qlsp-dissipative}\subref{fig:qlsp-steps}, the SPD2 evolution error
decreases from $2.98\times10^{-3}$ at $m=32$ to
$1.15\times10^{-5}$ at $m=512$, approximately by a factor of four
per doubling. The recovered-vector error approaches
$4.36\times10^{-5}$, showing that time-step refinement alone
does not remove the profile and spatial errors.
Panel~\subref{fig:qlsp-profile} examines profile and Fourier projection errors without
time-stepping error, using exact finite-Hamiltonian evolution with $N_p=256$ and
$(L,N_q)=(2,32),(4,64),(6,128),(8,256)$.
The relative vector error decreases from $5.27\times10^{-3}$ to $4.48\times10^{-6}$, while the normalized-state error
decreases from $1.31\times10^{-3}$ to $1.45\times10^{-6}$.
We also check the $p$-projection independently against the continuous
periodic convection solution. For SPD2, the maximum relative discrepancy
at nodes in $\mathcal R_p$ decreases from $1.20\times10^{-3}$ to
$3.56\times10^{-6}$ and $2.23\times10^{-15}$ as $N_p=64,128,256$.
Panel~\subref{fig:qlsp-solution} compares the real and imaginary solution components for
Asym2, including the magnitude restored from the circuit branch.

\begin{table}[tbp]
\centering\small
\caption{UnitaryLab circuit results with $r=8$, $L=6$, $N_q=128$,
and $m=512$. The probability includes both auxiliary projections and
upper-block selection. Gain denotes $\mathbb P_p/\mathbb P_{\rm point}$.}
\label{tab:qlsp-dissipative}
\begin{tabular}{lrrrrrr}
\hline
System & $N_p$ & Qubits & $E_{\rm vec}$ & $E_{\rm state}$ & $\mathbb P_p$ & Gain \\
\hline
SPD2 & 128 & 16 & $4.36\times10^{-5}$ & $1.37\times10^{-5}$ & $3.53\times10^{-3}$ & 9.00 \\
Asym2 & 256 & 17 & $4.28\times10^{-5}$ & $1.96\times10^{-5}$ & $1.90\times10^{-3}$ & 13.00 \\
Poisson4 & 512 & 19 & $5.01\times10^{-5}$ & $6.86\times10^{-6}$ & $4.03\times10^{-3}$ & 9.00 \\
\hline
\end{tabular}
\end{table}

\paragraph{Recovery probability.}
We compute probabilities from the circuit statevector before amplitude
amplification, without finite-shot sampling.
The predicted interval probability, including the dissipative shift
and the weighted $q$-recovery, is
\[
\mathbb P_p\approx
\frac{\e^{-1}M_p\|\bb x\|^2}
{T^2\|\bb\psi_{\Pi}\|^2
\|\Psi_{q,\Pi}\|_{L^2}^2\|\chi_{q,\Pi}\|_{L^2}^2}.
\]
Panel~\subref{fig:qlsp-probability} compares coherent interval recovery with selection of $p=0$
on the same evolved state, keeping the $q$-projection and upper-block
selection fixed. At $N_p=128$ and $256$, the interval probabilities are
$3.53\times10^{-3}$ and $3.73\times10^{-3}$, whereas the point
probability falls from $3.93\times10^{-4}$ to $1.96\times10^{-4}$.
The coarsest grid, $N_p=32$, has relative vector error
$7.71\times10^{-3}$; its probability alone does not indicate accurate recovery.
The gain is approximately $M_p$, illustrating why a fixed recovery
interval avoids the loss of probability caused by selecting a single
point on a refined grid. These experiments use the unpreconditioned
systems; amplitude amplification and solution-norm estimation are
not executed.

\begin{figure}[tbp]
\centering
\subfigure[Time-step refinement (SPD2).]{%
\includegraphics[width=0.48\textwidth]{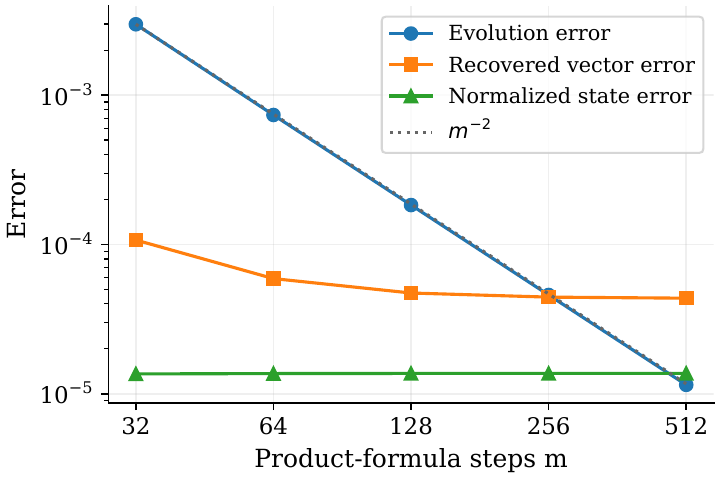}%
\label{fig:qlsp-steps}}\hfill
\subfigure[Periodic-profile refinement (SPD2).]{%
\includegraphics[width=0.48\textwidth]{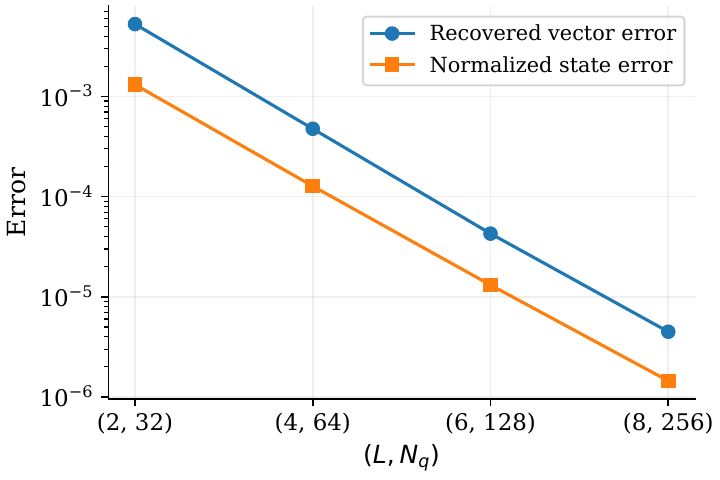}%
\label{fig:qlsp-profile}}\\
\subfigure[Recovery probabilities (SPD2).]{%
\includegraphics[width=0.48\textwidth]{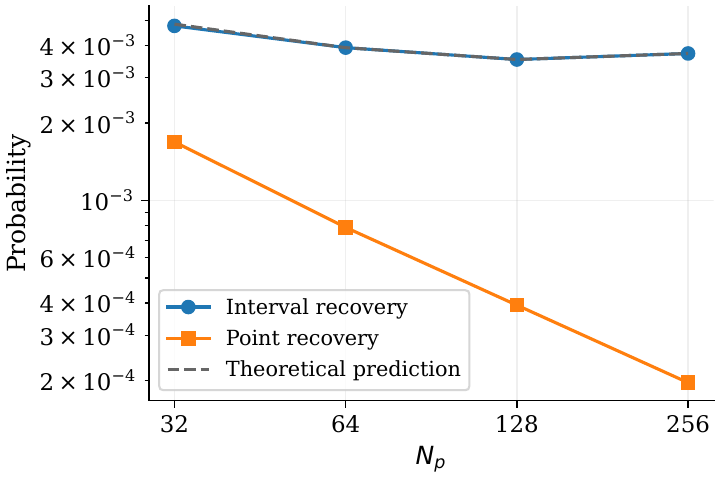}%
\label{fig:qlsp-probability}}\hfill
\subfigure[Solution components (Asym2).]{%
\includegraphics[width=0.48\textwidth]{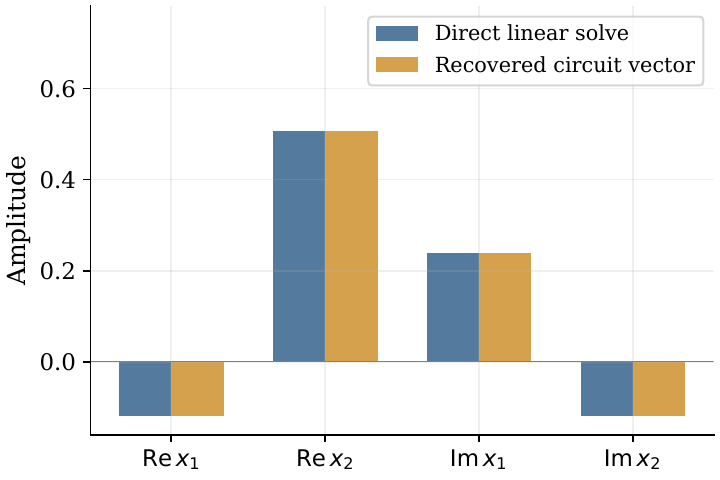}%
\label{fig:qlsp-solution}}
\caption{Validation of the dissipative Schr\"odingerization construction.
\subref{fig:qlsp-steps} Time-step refinement for SPD2 with $N_p=128$; the evolution error
is the norm difference from exact evolution of the same finite Hamiltonian.
\subref{fig:qlsp-profile} Periodic-profile refinement with exact finite-Hamiltonian evolution.
\subref{fig:qlsp-probability} Circuit recovery probabilities at $m=256$, with the other parameters fixed.
\subref{fig:qlsp-solution} Complex solution components for Asym2 at $m=512$.}
\label{fig:qlsp-dissipative}
\end{figure}
\FloatBarrier

\section{Conclusions}
We have realized Schr\"odingerization of the QLSP in two higher
dimensions. Starting from the time-integrated convection response
used in LC-Schr\"odingerization, we use Duhamel's principle to obtain
an inhomogeneous convection equation. Fourier projection and the
additional Schr\"odingerization variable then allow us to implement the
response without an LCU over evolution times.

The joint choice of the kernel and recovery procedure gives an
evolution time independent of the target accuracy and a uniformly
bounded $L^2$ kernel normalization. The resulting convection solution
approximates the target vector on a fixed interval, and coherent recovery
over that interval compensates for the sampled normalization.
Combining these estimates with the Schr\"odingerization solver
retains a single logarithmic precision factor. Under the stated
oracle assumptions and given a valid constant-factor solution-norm
estimate, block preconditioning yields
$\mathcal O(\kappa_A\log(1/\varepsilon))$ queries to each original
input oracle without VTAA, with constant success probability and
$\ell^2$ state error at most $\varepsilon$.

Our UnitaryLab experiments demonstrate recovery for positive-definite
and indefinite systems. They separate time-stepping error from
auxiliary-variable approximation errors and show the probability gain
from recovery over a fixed interval.

\bibliographystyle{unsrt} 
\bibliography{Refs2}
\end{document}